%% file: ipl2026.tex
\documentclass[a4paper,fleqn]{cas-dc}

\usepackage[authoryear,longnamesfirst]{natbib}

\input{packages}
\input{macros}

\begin{document}
\let\WriteBookmarks\relax
\def\floatpagepagefraction{1}
\def\textpagefraction{.001}

\shorttitle{Satisfiability for Knowing How over Linear Plans is NP-complete}    

\shortauthors{Areces, Barcel\'o, Cassano, Castro, Demri \& Fervari}  

\title [mode = title]{Satisfiability for Knowing How over Linear Plans is NP-complete}  



%

\author[1,2]{Carlos Areces}[orcid=0000-0001-7845-8503]

\author[3]{Pablo Barcel\'o}[orcid=0000-0003-2293-2653]

\author[1,4]{Valentin Cassano}[orcid=0000-0001-5904-3038]

\author[1,4]{Pablo F. Castro}[orcid=0000-0002-5835-4333]

\author[5]{St\'ephane Demri}[orcid=0000-0002-3493-2610]

\author[1,2]{Raul Fervari}[orcid=0000-0003-0360-0725]
\ead{rfervari@unc.edu.ar}
\cormark[1]

\affiliation[1]{organization={Consejo Nacional de Investigaciones Cient\'ificas y T\'ecnicas, CONICET},
                addressline={Godoy Cruz 2290}, 
                postcode={C1425FQB}, 
                postcodesep={}, 
                city={CABA},
                country={Argentina}}

\affiliation[2]{organization={Universidad Nacional de C\'ordoba, FAMAF},
                addressline={Medina Allende s/n}, 
                postcode={X5000HUA}, 
                postcodesep={}, 
                city={C\'ordoba},
                country={Argentina}}

\affiliation[3]{organization=Pontificia Universidad Cat\'olica de Chile \& IMFD \&  CENIA,
                 addressline={Campus San Joaquín, Vicuña Mackenna 4860}, 
                postcode={7820436}, 
                postcodesep={}, 
                city={Macul},
                country={Chile}}

\affiliation[4]{organization={Universidad Nacional de R\'io Cuarto},
                addressline={Ruta Nac. 36 - KM. 601}, 
                postcode={X5804BYA}, 
                postcodesep={}, 
                city={R\'io Cuarto},
                country={Argentina}}

\affiliation[5]{organization={Université Paris-Saclay, CNRS, ENS Paris-Saclay, Laboratoire Méthodes Formelles},
                addressline={4 Avenue des Sciences}, 
                postcode={91190}, 
                postcodesep={}, 
                city={Gif-Sur-Yvette},
                country={France}}







\cortext[1]{Corresponding author}



\begin{abstract}
We study the satisfiability problem for a modal logic expressing
\emph{knowing-how} assertions, which captures an agent’s ability to achieve
a given goal under the standard semantics based on linear plans.
Our main result shows that satisfiability of knowing-how formulas is
$\NP$-complete, improving previously known complexity bounds. 
The proof proceeds via a translation into modal logic S5, an 
instrumental tool for addressing a variety of problems in knowledge representation.
\end{abstract}





\maketitle


\section{Introduction}
\label{sec:intro}
\input{intro}

\section{Preliminaries}
\label{sec:basic}
\input{kh}

\section{The Satisfiability Problem of $\KHlogic$ is in NP}
 \label{sec:np}
\input{alternative}

\section{Final Remarks}
\label{sec:final}
\input{final}

\printcredits

\bibliographystyle{cas-model2-names}

\bibliography{references}



\end{document}

%% file: packages.tex
\usepackage[usenames,dvipsnames]{xcolor}
\usepackage{amsmath}
\usepackage{amssymb}
\usepackage{amsthm}
\usepackage{multicol}
\usepackage{xspace}
\usepackage[textsize=scriptsize]{todonotes}
\usepackage{algorithm}
\usepackage{soul}
\usepackage{mdframed}
\usepackage{flushend}
\usepackage[noEnd=true]{algpseudocodex}
  \algnewcommand{\IIf}[1]{\State\algorithmicif\ #1\ \algorithmicthen}
  \algnewcommand{\EndIIf}{\unskip\ \algorithmicend\ \algorithmicif}

  \makeatletter
  \renewcommand{\ALG@beginalgorithmic}{\scriptsize}
  \makeatother
  \algrenewcommand\alglinenumber[1]{\small #1:}

\usepackage{proof}

\usepackage[draft]{fixme}

\usepackage{stmaryrd}
\usepackage{paralist}

\usepackage{tikz}
\usetikzlibrary{positioning}

\usepackage{multirow}

\usepackage{xparse}

\usepackage[all]{xy}
\newdir{|>}{%
  !/4.5pt/@{|}*:(1,-.2)@^{>}*:(1,+.2)@_{>}}

\usepackage[english]{babel}

\usepackage{hyperref}
\hypersetup{
  linkcolor  = blue,
  citecolor  = blue,
  urlcolor   = blue,
  colorlinks = true,
}

\usepackage{mathtools}
\usepackage{forest}
\usepackage{etoolbox}

\usepackage{thm-restate}

%% file: macros.tex
\newtheorem{theorem}{Theorem}
\newtheorem{proposition}{Proposition}
\newtheorem{corollary}{Corollary}
\newtheorem{lemma}{Lemma}

\newtheorem{definition}{Definition}
\newtheorem{example}{Example}

\usepackage{cleveref}

\Crefname{algorithm}{Alg.}{Algs.}
\Crefname{definition}{Def.}{Defs.}
\Crefname{equation}{Eq.}{Eqs.}
\Crefname{figure}{Fig.}{Figs.}
\Crefname{proposition}{Prop.}{Props.}
\Crefname{theorem}{Thm.}{Thms.}
\Crefname{example}{Ex.}{Exs.}
\Crefname{corollary}{Cor.}{Cors.}
\Crefname{section}{Sec.}{Secs.}
\Crefname{appendix}{App.}{Apps.}

\usepackage{tikz}
\usetikzlibrary{arrows,decorations,shapes,automata,positioning,decorations.pathmorphing,patterns}

\tikzset{
  every picture/.style = {
    thick,
    ->,
    >=stealth',
    node distance = 1.5em and 3em,
  }
  ,
  cross line/.style = {
    preaction = {
      draw=white,
      -,
      line width=4pt
    }
  }
  , 
  state/.style = {
    circle,
    font = \footnotesize,
    draw,
    minimum width = 1em,
    minimum height = 1em
  }
  , 
  label-state/.style = {
    sloped,
    font = \scriptsize,
    label distance = -2pt
  }
  , 
  label-edge/.style = {
    font = \scriptsize,
    label distance = -2pt
  }
}

\newcommand{\PROP}{{\rm \sf Prop}\xspace}
\newcommand{\ACT}{{\rm \sf Act}\xspace}

\newcommand{\kh}{{\mathsf{Kh}}}

\newcommand{\liff}{\leftrightarrow}

\newcommand{\A}{{\operatorname{\sf A}}}
\newcommand{\E}{{\operatorname{\sf E}}}

\newcommand{\size}[1]{{|}{#1}{|}}

\newcommand{\KHlogic}{\ensuremath{\mathfrak{L}({\kh})}\xspace}

\newcommand{\Alogic}{\ensuremath{\mathfrak{L}({\A})}\xspace}
\newcommand{\SFive}{\Alogic}

\newcommand{\nmodel}{\mathfrak{N}}
\newcommand{\model}{\mathfrak{M}}

\newcommand{\R}{\operatorname{R}}
\renewcommand{\S}{\operatorname{S}}

\newcommand{\V}{\operatorname{V}}

\newcommand{\plan}{\pi}

\newcommand{\PLANS}{\ACT^{*}}

\newcommand{\lts}{\textup{LTS}\xspace}

\newcommand{\truthset}[2]{\llbracket #2 \rrbracket^{#1}}

\newcommand{\stexec}{\operatorname{SE}}

\newcommand{\card}[1]{{\mid}{#1}{\mid}}
\newcommand{\tup}[1]{\langle{#1}\rangle}

\newcommand{\setof}[2]{\{{#1}\mid {#2}\}}
\newcommand{\set}[1]{\{{#1}\}}

\newcommand{\colornuevo}{teal}

\newcommand{\colornota}{Peach}

\newcommand{\colorincompleto}{red}

\newcommand{\NP}{{\mathsf{NP}}\xspace}

\newcommand{\Poly}{{\mathsf{P}}\xspace}
\newcommand{\PSPACE}{{\mathsf{PSpace}}\xspace}

\newcommand{\zerodisplayskips}{%
  \setlength{\abovedisplayskip}{5pt}%
  \setlength{\belowdisplayskip}{5pt}%
  \setlength{\abovedisplayshortskip}{5pt}%
  \setlength{\belowdisplayshortskip}{5pt}}
\appto{\normalsize}{\zerodisplayskips}
\appto{\small}{\zerodisplayskips}
\appto{\footnotesize}{\zerodisplayskips}

\DeclareMathOperator{\depth}{\mathsf{md}}

\newcommand{\C}[1]{{#1{}}^{*}}

\newcommand{\aformula}{\varphi}
\newcommand{\aformulabis}{\psi}
\newcommand{\aformulater}{\chi}

\newcounter{ieqnnum}
\makeatletter
\newcommand*{\ieqn}[2][]{%
    \begingroup
        \refstepcounter{ieqnnum}%
        (\roman{ieqnnum})\xspace%
        \ifx\\#1\\%
        \else
            \label[ieqn]{#1}
        \fi
        \relpenalty=10000 %
        \binoppenalty=10000 %
        \ensuremath{%
            #2%
        }%
    \endgroup
}
\makeatother

\newcounter{rieqnnum}%
\newcommand{\ieqnref}[1]{%
    \setcounter{rieqnnum}{#1}%
    \setcounter{rieqnnum}{%
        \numexpr\value{ieqnnum}+\value{rieqnnum}%
    }%
    (\roman{rieqnnum})%
}

\newcommand{\defstyle}[1]{\emph{#1}}

\newcommand{\cut}[1]{} 
\newcommand{\Rfamily}{(\R_a)_{a \in \ACT}}

\newcommand{\pair}[2]{(#1,#2)}

%% file: intro.tex

The notion of \emph{knowing-how} captures an agent’s ability to achieve a goal
(see~\cite{fantl2012introduction,Pavese22}) and offers a natural and simple way to reason about
strategies in AI, including automated planning (see, e.g.,~\cite{HBEL}).
The framework introduced in~\cite{Wang15lori,Wang2016} has since become a
standard formalization of this idea.
It extends propositional logic with a binary modality $\kh(\varphi,\psi)$,
interpreted over labelled transition systems (LTSs), expressing that “whenever
$\varphi$ holds, the agent knows how to achieve~$\psi$.”

Semantically, an agent can achieve $\psi$ given $\varphi$ (i.e., $\kh(\varphi,\psi)$ holds) when there exists a
linear plan that is always executable from the states satisfying $\varphi$, it never aborts, and
always reaches states satisfying $\psi$.
A major appeal of the logic is its simplicity: the modality $\kh$ is global and
does not refer to actions or plans explicitly.
For this reason, the proposal from~\cite{Wang15lori,Wang2016} has become
foundational, inspiring several subsequent investigations and variants of knowing-how (see, e.g.,~
\cite{FervariHLW17,Li17,Naumov2018a,ArecesFSV25}).

\paragraph{Contributions.}
In this paper, we establish an optimal upper bound for the satisfiability problem of the
knowing-how logic introduced in~\cite{Wang15lori,Wang2016}, which we denote by
$\KHlogic$. More precisely, we show that this problem is $\NP$-complete,
thus refining the best previously known upper bound presented in~\cite{ACCFS23}. In
that work, satisfiability for $\KHlogic$ is decided by an $\NP^\NP$ algorithm
(also denoted $\Sigma^{\Poly}_2$), corresponding to the second level of the
polynomial hierarchy~\citep{Stock76}. Our result closes this gap and yields, as
a direct consequence, a polynomial-size model property for $\KHlogic$: whenever
a formula $\varphi$ is satisfiable, it is satisfiable over an LTS of polynomial size with respect to the size of $\varphi$.

Technically, we reduce satisfiability of $\KHlogic$ formulas to satisfiability
in propositional logic extended with global universal 
modalities~\citep{GorankoP92,Hemaspaandra96}, which we denote by $\Alogic$.
Our analysis is clearly based on~\cite{ACCFS23}, for which we provide several substantial refinements to get the optimal $\NP$ upper bound. Precisely, herein we manage to directly translate $\KHlogic$ formulas into equisatisfiable $\Alogic$ formulas, thereby eliminating the unsatisfiability checks used in~\cite{ACCFS23}, which prevented us from obtaining an optimal bound.  
The logic $\Alogic$ is a syntactic variant of S5 (see~\cite{mlbook}), whose
satisfiability problem is $\NP$-complete. S5 has been used in several
knowledge-representation tasks, including knowledge
compilation~\citep{Bienvenu&Fargier&Marquis10} and belief
revision~\citep{AVBNH25}, although it has been also studied in its own right
\citep{Salhi&Sioutis15,Niveau&Zanuttini16}.

Thus, we obtain tight complexity bounds for a logic that has recently received
wide attention. Since $\KHlogic$ is strictly more expressive than $\Alogic$ over LTSs,
this extends the frontier of logics with $\NP$-complete satisfiability.
Moreover, model checking for $\KHlogic$ is $\PSPACE$-complete~\citep{DF23}, and
hence potentially harder than satisfiability, a truly uncommon phenomenon.

\paragraph{Outline.} 
The article is organized as follows. 
In~\Cref{sec:basic} we introduce some notations and the basic definitions of $\KHlogic$ and $\Alogic$. \Cref{sec:np} is devoted to incrementally show our complexity result. 
Finally, in~\Cref{sec:final} we provide some remarks and future lines of research.


%% file: kh.tex
\subsection{The logic $\KHlogic$}

In this section, we present the syntax, the models and the semantics 
of the logic of \emph{Knowing How (over Linear Plans)}, which we denote by~$\KHlogic$, introduced first in~\cite{Wang15lori}. 

\paragraph{{\bf Syntax.}}
We assume that $\PROP$ and $\ACT$ are countable sets of \emph{propositional symbols} and \emph{action symbols}, respectively.

\begin{definition}\label{def:khsyntax}
    The set of \emph{formulas} of $\KHlogic$ is defined recursively by the following grammar: 
    \begin{align*}
        \varphi,\psi & \Coloneqq
            p \mid
            {\lnot \varphi} \mid
            {\varphi \lor \psi} \mid
            \kh(\varphi,\psi),
    \end{align*}
    where $p\in\PROP$. 
    We use $\bot$, $\top$, $\varphi \land \psi$, $\varphi \to \psi$, and ${\varphi \liff \psi}$ with their standard meanings.
     \qed
\end{definition}

Formulas of the form $\kh(\varphi,\psi)$, called \emph{knowing-how},  intuitively indicate that ``the agent knows how to achieve the \emph{postcondition} $\psi$ given a \emph{precondition} $\varphi$''.

\paragraph{{\bf Semantics.}}
The formulas of $\KHlogic$ are interpreted over \emph{labelled transition systems} built on a fixed set $\ACT$ of the \emph{basic action} symbols. Each basic action corresponds to a single step that moves the system from one state to another.

\begin{definition}\label{def:lts}
    A \emph{labelled transition system (\lts)}, often referred to as an {\em $\KHlogic$ model}, is a tuple $\model = 
    \tup{\S,\Rfamily,\V}$ such that:
   (1) $\S$ is a non-empty set of \emph{states};
    (2) $\Rfamily$ is a family of binary relations on $\S$; and
    (3) $\V: \PROP \to 2^{\S}$ is a \emph{valuation function}.
   \qed
\end{definition}

A characteristic of knowing-how formulas is that they capture an agent's ability to carry out, not an individual basic action, but a whole sequence of basic actions that leads to success. Such a sequence is called a \emph{plan}.
Simply put, a plan is a finite sequence of basic actions, describing how the agent acts step by step. We write $\PLANS$ for the set of all plans. Moreover, we denote by $\varepsilon$ the empty plan, which represents the option of doing nothing at all. Plans give rise to new relations built on those in a model. These new relations are defined below.

\begin{definition}\label{def:plans-aux}
    Let 
    $\Rfamily$ be a  family of binary relations on $\S$.
    We define $\R_{\varepsilon} = \setof{(s,s)}{s \in \S}$.
    In turn, for every $\plan \in \PLANS$, and $a \in \ACT$, we define
    $\R_{\plan{a}}$ as the (relational) composition of $\R_{\plan}$ and $\R_{a}$.
    Finally, for all $\R_\plan$,
        and all $A\subseteq\S$, we define $\R_\plan(A)=\setof{t}{\text{there is } s \in A \mbox{ s.t. } (s,t)\in\R_\plan}$,
        and $\R_\plan(t)=\R_\plan(\set{t})$. \qed
\end{definition}

The distinguishing feature of $\KHlogic$ is the notion of \emph{strong executability}, which characterizes when a plan is \emph{adequate} for execution. This notion is inspired by conformant planning in~\cite{Smith&Weld98,CimattiPRT03}, and its motivation and significance are discussed in detail in~\cite{Wang15lori}. Intuitively, a plan is \emph{strongly executable} if, once its execution begins, it can always be carried through to completion: at every step of the plan, the next action is guaranteed to be performable, regardless of how the execution unfolds.

\begin{definition} \label{def:plans-exec}
    For any $\plan=a_1 \dots a_n \in \ACT^*$, and $1 \leq i \leq j \leq n$, we denote: 
    $\plan_i = a_i$;
    $\plan[i{:}j] = a_i \dots a_j$; and
    $|\plan| = n$.
    Let $\model = \tup{\S, \Rfamily, \V}$ be an $\KHlogic$ model.
    We say that a plan $\plan$ is \defstyle{strongly executable (SE)} at
    $s \in \S$, iff for all $0 \leq i < \card{\plan}$ and all $t \in \R_{(\plan[1{:}i])}(s)$, it follows that $\R_{\plan_{(i+1)}}(t) \neq \emptyset$. 
    The set of all states at which $\plan$ is strongly executable is defined as $\stexec(\plan) = \setof{s}{\plan \ \mbox{\rm is SE at } s}$.
    \qed
\end{definition}

Note $\stexec(\varepsilon) = \S$. 

To build intuition, we illustrate the strong executability of a plan and its execution in an $\KHlogic$ model with an example.

\begin{example}\label{ex:rel-notions}
    Consider the following $\KHlogic$ model:
    \begin{center}
        \begin{tikzpicture}[->]
            \node [state,label=above:\small$s$] (w1) {$p$};
            \node [state, label=above:\small$t$,right = of w1] (w2) {$r$};
            \node [state, label=above:\small$v$, below = of w2] (w3) {$r$};
            \node [state, label=above:\small$u$,right = of w2] (w4) {$q$};


            \path (w1) edge node [label-edge, above] {$a$} (w2)
                (w2) edge node [label-edge, above] {$b$} (w4)
                (w1) edge node [label-edge, above] {$a$} (w3);

        \end{tikzpicture}
    \end{center}
    Moreover, consider the plan
    $ab$.
    We have $\R_a(s)=\{t,v\}$ and $\R_{ab}(s)=\{u\}$.
    Thus, $a$ is strongly executable at $s$, i.e., $s\in\stexec(a)$.
    However, $ab$ is not, i.e., $s\notin\stexec(ab)$.
    This is the case since $v\in\R_{a}(s)$ but $\R_b(v)=\emptyset$.
    In summary, $\stexec(\varepsilon)=\S$, $\stexec(a)=\{s\}$, and
    $\stexec(ab)=\emptyset$.
\end{example}

With the key semantic ingredients in place, we are ready to  define how formulas of $\KHlogic$ are evaluated on $\KHlogic$ models. Our presentation follows~\cite{Wang15lori,Wang2016}.

\begin{definition} \label{def:khsemantics}
    Let $\model = \tup{\S,\Rfamily,\V}$ be an $\KHlogic$ model.
    We define $\truthset{\model}{\varphi}$ recursively as:
    \[
    \begin{array}{c}
        \truthset{\model}{p} {=}
            \V(p)   \quad  \quad 
        \truthset{\model}{\lnot \varphi} {=}
            \S \setminus \truthset{\model}{\varphi}  \\  
         \truthset{\model}{\varphi \lor \psi} {=}
            \truthset{\model}{\varphi} \cup \truthset{\model}{\varphi} \\
       \truthset{\model}{\kh(\varphi,\psi)} =
            \begin{cases}
                \S & \text{if there is } \plan{\in}\PLANS \text{ s.t. } 
                        \truthset{\model}{\varphi} \subseteq \stexec(\plan) \\ 
                      &   \text{ and } \R_{\plan}(\truthset{\model}{\varphi}) \subseteq \truthset{\model}{\psi} \\
                \emptyset & \text{otherwise}. \\
            \end{cases}
    \end{array}
    \]
   We say that a plan $\plan \in \PLANS$ is a \emph{witness} for $\kh(\varphi,\psi)$ iff 
       $\truthset{\model}{\varphi} \subseteq \stexec(\plan)$ and
       $\R_{\plan}(\truthset{\model}{\varphi}) \subseteq \truthset{\model}{\psi}$. \qed
\end{definition}

\begin{example}\label{exs:semantics}
Let $\model$ be the $\KHlogic$ model from~\Cref{ex:rel-notions}. It is easy to verify ${\truthset{\model}{\kh(p,r)}=\S}$ (using $a$ as a witness), while $\truthset{\model}{\kh(p,q)}=\emptyset$ (there is no witness for the formula).
\end{example}

\subsection{The logic $\Alogic$}

The logic $\Alogic$ is obtained from the grammar in~\Cref{def:khsyntax} by
replacing $\kh(\varphi,\psi)$ with $\A\varphi$.
The modality $\A\varphi$, called \emph{universal}, holds at all states of an
$\KHlogic$ model $\model = 
    \tup{\S,\Rfamily,\V}$ iff $\varphi$ holds at every state $s \in S$.
The dual of $\A\varphi$ is defined as $\E\varphi = \neg\A\neg\varphi$, and is called \emph{existential}.
Then, observe  that $\A\varphi$ is equivalent to the $\KHlogic$ formula
$\kh(\lnot\varphi,\bot)$. This makes $\Alogic$ a semantic fragment of $\KHlogic$. As we show next, $\Alogic$ is strictly contained in $\KHlogic$ over $\KHlogic$ models.

\begin{theorem}
    \label{th:expressivity}
    $\KHlogic$ is strictly more expressive than $\Alogic$ over $\KHlogic$ models. 
\end{theorem}

\begin{proof}
We show that the formula $\kh(p,q)$ is not expressible in 
$\Alogic$. 
Consider the two $\KHlogic$ models 
$\model$ and $\model'$ below:
\begin{center}
        \begin{tikzpicture}[->]
        \node [state,label=above:\small$s$] (w1) {$p$};
        \node [state, label=above:\small$t$,right = of w1] (w2) {$q$};
        \node[below left of=w1,xshift=-0.3cm,yshift=0.3cm] (m) {$\model$};
    
        \node [state, label=above:\small$s'$, right = 2cm of w2] (w3) {$p$};
        \node [state, label=above:\small$t'$,right = of w3] (w4) {$q$};
        \node[below left of=w3,xshift=-0.3cm,yshift=0.3cm] (m) {$\model'$};


        \path (w1) edge node [label-edge, above] {$a$} (w2);

    \end{tikzpicture}
\end{center}
Clearly, these models satisfy the same $\Alogic$ formulas. However, we have $\truthset{\model}{\kh(p,q)}\neq\emptyset$ while $\truthset{\model'}{\kh(p,q)}=\emptyset$.
\end{proof}

Since the semantics of $\Alogic$ does not refer to actions, we can simply write its models as pairs $\model=\tup{\S,\V}$, where $\S$ is a set of states and $\V : \PROP \to 2^{\S}$ is a valuation function.
Such pairs are called {\em $\Alogic$ models}.

\subsection{The satisfiability problem}

Let $\cal L$ be $\KHlogic$ or $\Alogic$. The
\emph{satisfiability problem for $\cal L$}, written ${\sf SAT}({\cal L})$, is the following decision problem:

\begin{mdframed}[linewidth=0.8pt,roundcorner=4pt]
${\sf SAT}({\cal L})$ 
\hrule\vspace{0ex}
\begin{description}\setlength\itemsep{0em}
  \item[Input:] A formula $\varphi$ of $\cal L$.
  \item[Question:] Is there an $\cal L$ model $\model$ with
  $\truthset{\model}{\varphi} \neq \emptyset$?
\end{description}
\end{mdframed}

The main focus of this paper is the computational complexity of the
satisfiability problem for $\KHlogic$.
To obtain a complete characterization, we rely on the following classical
result concerning satisfiability in $\Alogic$.


\begin{proposition}
{\em \cite[Th.~6.2]{Ladner77}} 
    ${\sf SAT}(\Alogic)$ is $\NP$-complete.
\end{proposition}

\subsection{Subjective formulas}

Finally, we identify and introduce a fragment of $\KHlogic$ that is sufficient for establishing our results. A \defstyle{positive atom} (resp. \defstyle{negative atom}) is an $\KHlogic$ formula of the form
$\kh(\aformulabis, \aformulater)$ (resp. $\neg \kh(\aformulabis, \aformulater)$) such that $\aformulabis$ and $\aformulater$ are $\kh$-free (the same criteria also applies to formulas $\A\psi$ and $\E\psi$, since they are defined as abbreviations). 
A \defstyle{subjective formula} is a Boolean combination of atoms, see e.g.~\cite{Bienvenu&Fargier&Marquis10,Niveau&Zanuttini16}.
The following lemma can easily be 
shown witnessing the globality of positive atoms.

\begin{lemma}\label{lemma-subjective-fragment}
Let $\aformula$ be a subjective formula and $\model$ be an $\KHlogic$ model.
Then either $\truthset{\model}{\aformula} = \S$ or
$\truthset{\model}{\aformula} = \emptyset$.
\end{lemma}

\cut{
A formula $\aformula$ has \defstyle{modal depth at most one} iff its $\kh$-subformulae are positive atoms (no imbrication
of $\kh$-subformulae). 

\begin{lemma}\label{lemma-modal-depth-one}
For every $\KHlogic$ formula $\aformula$, there is an $\KHlogic$ formula $\aformula'$
with modal depth at most one, such that $\aformula$ and $\aformula'$ are logically equivalent.
\end{lemma}
}

\cut{The proof hinges on the following property.
Let $\kh(\aformulabis,\aformulater)$ be a positive atom occurring in $\aformula$
within the scope of a larger $\kh$-subformula.
Then $\aformula$ is logically equivalent to
\begin{multline*}
\big(\aformula[\kh(\aformulabis,\aformulater)\leftarrow\top]\wedge
\kh(\aformulabis,\aformulater)\big)
\;\wedge\\
\big(\aformula[\kh(\aformulabis,\aformulater)\leftarrow\perp]\wedge
\neg\kh(\aformulabis,\aformulater)\big),
\end{multline*}
where $\aformula[\kh(\aformulabis,\aformulater)\leftarrow\top]$
(resp.\ $\perp$) denotes the formula obtained by replacing all occurrences of
$\kh(\aformulabis,\aformulater)$ with $\top$ (resp.\ $\perp$).
Since this transformation strictly reduces the number of positive atoms under
larger $\kh$-subformulas, Lemma~\ref{lemma-modal-depth-one} follows.}

%% file: alternative.tex
In this section, we pinpoint the exact complexity of the satisfiability problem
for $\KHlogic$ by proving that it is in~$\NP$, thereby improving the
$\NP^\NP$ upper bound from~\cite{ACCFS23} and closing the complexity gap.
Our analysis takes a different route from~\cite{ACCFS23}: instead of relying on
unsatisfiability checks, we reduce $\KHlogic$ satisfiability directly to
$\Alogic$ satisfiability.
To streamline exposition, we begin by assuming that the formula~$\varphi$ is a
conjunction of atoms.
We return to this assumption and explain how to remove it in the general
algorithm (see the proof of \Cref{th:satkh-np}).

The overall strategy is to decompose the satisfiability check for~$\varphi$ into three subproblems, each addressing a different structural component of the formula.
\begin{enumerate}[(a)]
    \item We check the satisfiability of a conjunction of positive 
    atoms 
    of the form $\kh(\psi_1,\chi_1) \land \dots \land \kh(\psi_n,\chi_n)$.
    \item We check the satisfiability of conjunctions of 
    negative atoms
    of the form $\lnot\kh(\psi'_1,\chi'_1) \land \dots \land \lnot\kh(\psi'_m,\chi'_m)$.
    \item We combine the results from (a) and (b) to decide the satisfiability of conjunctions of atoms, which shall be sufficient to handle the general case.
\end{enumerate}
The next subsections address each subproblem in turn.

\subsection{Satisfiability for conjunctions of positive atoms}
\label{section-satplus}
\input{satplus}

\subsection{Satisfiability for conjunctions of negative atoms}
\label{section-satminus}
\input{satminus}

\subsection{Satisfiability for conjunctions of  atoms}
\label{section-satflat}
\input{satflat}

%% file: satplus.tex
We first characterize the satisfiability status of a conjunction of positive
atoms of the form $\varphi^{+}=\kh(\psi_1,\chi_1) \wedge \cdots \wedge \kh(\psi_n,\chi_n)$.
In~\cite[Prop.~5]{ACCFS23}, this subproblem relies on \emph{unsatisfiability} checks of propositional formulas. Here, instead, we simply reduce the satisfiability of $\varphi^{+}$ to the satisfiability of an $\Alogic$ formula. The latter states that for each conjunct $\kh(\psi_i,\chi_i)$, either its precondition $\psi_i$ is false everywhere in the $\KHlogic$ model (i.e., $\A\neg\psi_i$ holds), or its postcondition~$\chi_i$ holds somewhere~(i.e., $\E\chi_i$ holds).

\begin{lemma}\label{lemma:satplus}
\cut{
    Let~$I$ be a finite set of indices.
    Define the formulas~$\varphi^{+}$ and~$\theta^{+}$ as follows:
    \[
      {\varphi^{+}} =  {
        \bigwedge_{i \in I} \kh(\psi_i,\chi_i)
        }
    \quad
    \mbox{and}
    \quad
        {\theta^{+}} = {
            \bigwedge_{i \in I}
            (\A(\lnot \psi_i) \lor \E\chi_i)
        }. 
    \]
    Let $\depth(\varphi^{+}) = 1$, 
    then, $\varphi^{+}$ is satisfiable iff $\theta^{+}$ is satisfiable.
}
Let $\varphi^{+}$ 
and $\theta^{+}$ 
be subjective formulas: 
\[
      {\varphi^{+}} =  {
        \bigwedge_{i \in I} \kh(\psi_i,\chi_i)
        }
    \quad
    \mbox{and}
    \quad
        {\theta^{+}} = {
            \bigwedge_{i \in I}
            (\A(\lnot \psi_i) \lor \E\chi_i)
        }. 
    \]
Then, $\varphi^{+}$ is satisfiable iff $\theta^{+}$ is satisfiable.
\end{lemma}
\begin{proof}
    \setcounter{ieqnnum}{0}%
    $(\Rightarrow)$ Suppose that $\varphi^{+}$ is satisfiable.
    There is an $\KHlogic$ model $\model = \tup{\S,\Rfamily,\V}$ 
    with
        \ieqn{\truthset{\model}{\varphi^{+}} = \S}.
    Based on $\model$, we build an $\SFive$ model $\nmodel = \tup{\S, \V}$ and prove that for all $i \in I$, $\truthset{\nmodel}{\A(\lnot \psi_i) \lor \E\chi_i} = \S$. 

    From \ieqnref0, we know that for all $i \in I$, it holds that 
        \ieqn{\truthset{\model}{\kh(\psi_i,\chi_i)} = \S}.
    From \ieqnref0, we know there is a witness plan $\pi \in \PLANS$ such that
        \ieqn{\truthset{\model}{\psi_i} \subseteq \stexec(\pi)} and
        \ieqn{\R_{\pi}(\truthset{\model}{\psi_i}) \subseteq \truthset{\model}{\chi_i}}.
    We proceed with a case analysis. 

    \begin{description}
        \item[Case $\truthset{\model}{\chi_i} \neq \emptyset$:] Then,
        we have $\truthset{\nmodel}{\chi_i} \neq \emptyset$.
        This ensures $\truthset{\nmodel}{\E\chi_i} = \S$, which yields $\truthset{\nmodel}{\A(\lnot \psi_i) \lor \E\chi_i} = \S$. 

        \item[Case $\truthset{\model}{\chi_i} = \emptyset$:] 
        From \ieqnref0, we have $\R_{\pi}(\truthset{\model}{\psi_i}) \subseteq \emptyset$.
          The claim is: 
              \ieqn{\truthset{\nmodel}{\psi_i} = \emptyset}.
         To prove the claim, suppose by contradiction that there is $s \in \S \cap \truthset{\nmodel}{\psi_i}$.
        Since $\pi$ is strongly executable at $s$, we have $\R_{\pi}(\truthset{\model}{\psi_i}) \neq \emptyset$.
         But this contradicts the fact that $\pi$ was chosen as a witness for $\kh(\psi_i,\chi_i)$.
        Now, \ieqnref0 ensures $\truthset{\nmodel}{\A(\lnot\psi_i)} = \S$, which yields $\truthset{\nmodel}{\A(\lnot \psi_i) \lor \E\chi_i} = \S$. 
    \end{description}

    \noindent Hence, the satisfiability of $\varphi^{+}$ implies the satisfiability of $\theta^+$.

    \noindent
    ($\Leftarrow$) Suppose that $\theta^{+}$ is satisfiable.
    Then, there exists an $\SFive$ model $\nmodel = \tup{\S, \V}$ such that for all $i \in I$,
        \ieqn{\truthset{\nmodel}{\A(\lnot \psi_i) \lor \E\chi_i} = \S}.
    Given~$\nmodel$, we construct an $\KHlogic$ model~$\model = \tup{\S, \Rfamily, \V}$ with the goal of ensuring~${\truthset{\model}{\varphi^{+}} = \S}$. The main challenge in this construction lies in defining the transition relations~$\Rfamily$. To address this, we define a set
    \[
        I^{\star} = \setof{\ell \in I}{ \truthset{\nmodel}{\chi_\ell} \neq \emptyset}.
    \]
    Then, 
    define for all $a_{\ell} \in \ACT$: 
    \begin{align*}
        \R_{a_{\ell}} &=
            \begin{cases}
                {\truthset{\nmodel}{\psi_{\ell}} \times \truthset{\nmodel}{\chi_{\ell}}}
                    & \text{if } \ell \in I^{\star} \\
                \emptyset
                    & \text{if } \ell \notin I^{\star}.
            \end{cases}
    \end{align*}

    The claim is that for all $i \in I$, $\truthset{\model}{\kh(\psi_i,\chi_i)} = \S$.
    Let $i \in I$. We proceed with a case analysis.

    \begin{description}
        \item[Case $\truthset{\model}{\psi_i} = \emptyset$:] 
        By the semantics, we directly get that  $\truthset{\model}{\kh(\psi_i,\chi_i)} = \S$ (witness: empty plan).

        \item[Case $\truthset{\model}{\psi_i} \neq \emptyset$:] 
         Here we have $\truthset{\nmodel}{\A(\lnot\psi_i)} = \emptyset$, and so, from \ieqnref0, we have $\truthset{\nmodel}{\chi_i} \neq \emptyset$. This implies $i \in I^{\star}$. From this last fact, we can choose $a_i \in \ACT \subseteq \PLANS$ as a witness plan. By definition, we know that for this plan $\truthset{\model}{\psi_i} \subseteq \stexec(a_i)$ and $\R_{a_i}(\truthset{\model}{\psi_i}) \subseteq \truthset{\model}{\chi_i}$. This ensures that $\truthset{\model}{\kh(\psi_i,\chi_i)} = \S$. 
    \end{description}
     These two cases  establish that the satisfiability of $\varphi^{+}$ is implied by the satisfiability of the $\SFive$ formula $\theta^{+}$.
\end{proof}

%% file: satminus.tex
Now we turn our attention to the subproblem of checking the satisfiability status  
of a conjunction of negative atoms of the form  $\varphi^{-}=\neg\kh(\psi_1,\chi_1)\wedge\cdots\ \wedge \neg \kh(\psi_n,\chi_n)$.
This subproblem is solved similarly to the same case in~\cite[Prop.~7]{ACCFS23}, where for each atom $\neg\kh(\psi_i,\chi_i)$, a propositional satisfiability check of $\psi_i\wedge\neg\chi_i$ is performed. Here, we slightly simplify this, again by reducing the satisfiability of each negative atom to the satisfiability of the existential formula $\E(\psi_i\wedge\neg\chi_i)$. 

\begin{lemma}\label{lemma:satminus}
\cut{
    Let~$I$ be a finite set of indices.
    Define the formulas~$\varphi^{-}$ and~$\theta^{-}$ as follows:
    \[
        {\varphi^{-}} = {
        \bigwedge_{i \in I} \lnot\kh(\psi_i,\chi_i)
        } 
    \quad
    \mbox{and}
    \quad
       {\theta^{-}} =  {
            \bigwedge_{i \in I}
            \E(\psi_i \land \lnot\chi_i)
        }. 
    \]
    Let $\depth(\varphi^{-}) = 1$, 
    then, $\varphi^{-}$ is satisfiable iff $\theta^{-}$ is satisfiable.
    }
Let $\varphi^{-}$ and $\theta^{-}$ 
be subjective formulas:
\[ {\varphi^{-}} = {
        \bigwedge_{i \in I} \lnot\kh(\psi_i,\chi_i)
        } 
    \quad
    \mbox{and}
    \quad
       {\theta^{-}} =  {
            \bigwedge_{i \in I}
            \E(\psi_i \land \lnot\chi_i)
        }. 
    \]
    Then, $\varphi^{-}$ is satisfiable iff $\theta^{-}$ is satisfiable.
\end{lemma}
\begin{proof}
    \setcounter{ieqnnum}{0}%
    $(\Rightarrow)$ Suppose that $\varphi^{-}$ is satisfiable.
    Then, there is an $\KHlogic$  model $\model$ such that for all $i \in I$, $\truthset{\model}{\lnot\kh(\psi_i,\chi_i)} = \S$. Thus, 
    \ieqn{\truthset{\model}{\kh(\psi_i,\chi_i)} = \emptyset}.
    We show that the $\SFive$ model $\nmodel = \tup{\S, \V}$ satisfies that, 
    for all $i \in I$, $\truthset{\nmodel}{\E(\psi_i \land \lnot\chi_i)} = \S$, or equivalently, that $\truthset{\nmodel}{\psi_i \land \lnot \chi_i} \neq \emptyset$.

    Choose an arbitrary $i \in I$. From \ieqnref0, we have that for all $\pi \in \PLANS$, either
    $\truthset{\model}{\psi_i} \nsubseteq \stexec(\pi)$ or
    $\R_{\pi}(\truthset{\model}{\psi_i}) \nsubseteq \truthset{\model}{\chi_i}$.
    In particular, we have that $\varepsilon \in \PLANS$ and  
    \ieqn{\R_{\varepsilon}(\truthset{\model}{\psi_i}) \nsubseteq \truthset{\model}{\chi_i}}.
    By definition, $\R_{\varepsilon}(\truthset{\model}{\psi_i}) = \truthset{\model}{\psi_i}$.
    Together with \ieqnref0, this yields $\truthset{\model}{\psi_i} \nsubseteq \truthset{\model}{\chi_i}$.

    Then, there is $s \in \S$ such that $s \in \truthset{\model}{\psi_i}$ and $s \notin \truthset{\model}{\chi_i}$. Equivalently, $s \in \truthset{\model}{\psi_i \land \lnot \chi_i}$.
    This last step ensures $\truthset{\nmodel}{\psi_i \land \lnot \chi_i} \neq \emptyset$, and therefore, $\truthset{\nmodel}{\E(\psi_i \land \lnot \chi_i)} = \S$.

    Thus, the satisfiability of $\varphi^{-}$ implies the satisfiability of 
    the $\SFive$ formula $\theta^{-}$.

    \noindent
    $(\Leftarrow)$ Suppose that $\theta^{-}$ is satisfiable. Then, there is an $\SFive$ model $\nmodel = \tup{\S, \V}$ such that for all $i \in I$, $\truthset{\nmodel}{\E(\psi_i \land \lnot\chi_i)} = \S$, or equivalently,
    \ieqn{\truthset{\nmodel}{\psi_i \land \lnot\chi_i} \neq \emptyset}.

    From~$\nmodel$ we construct a model~$\model = \tup{\S, \Rfamily, \V}$ in which for all $a \in \ACT$,
    we set $\R_{a} = \emptyset$. The claim is that for all $i \in I$, $\truthset{\model}{\kh(\psi_i,\chi_i)} = \emptyset$, and so $\truthset{\model}{\lnot\kh(\psi_i,\chi_i)} = \S$.
    To prove this claim, we choose an arbitrary $i \in I$ and an arbitrary $\pi \in \PLANS$. We proceed with a case analysis.

    \begin{description}
            \item[Case $\pi \neq \varepsilon$:] 
        By definition of $\model$, we have $\stexec(\pi) = \emptyset$.
        Moreover, from \ieqnref0, $\truthset{\model}{\psi_i} \neq \emptyset$.
        This implies $\truthset{\model}{\psi_i} \nsubseteq \stexec(\pi)$.
        So, $\pi$ cannot be a witness plan for $\kh(\psi_i,\chi_i)$.

        \item[Case $\pi = \varepsilon$:] 
        From \ieqnref0, we have $\truthset{\model}{\psi_i} \nsubseteq \truthset{\model}{\chi_i}$.
        This implies $\R_{\varepsilon}(\truthset{\model}{\psi_i}) \nsubseteq \truthset{\model}{\chi_i}$.
        So, $\varepsilon$ cannot be used as a witness plan for $\kh(\psi_i,\chi_i)$ either.
    \end{description}
    These two cases establish that for all plan $\pi \in \PLANS$, $\truthset{\model}{\psi_i} \nsubseteq \stexec(\pi)$ or ${\R_{\pi}(\truthset{\model}{\psi_i}) \nsubseteq \truthset{\model}{\chi_i}}$, or equivalently $\truthset{\model}{\kh(\psi_i,\chi_i)} = \emptyset$. Hence, we obtain $\truthset{\model}{\lnot\kh(\psi_i,\chi_i)} = \S$. 
    
    Thus, satisfiability of $\theta^{-}$ implies satisfiability of $\varphi^{-}$.
\end{proof}

%% file: satflat.tex
In this section,  we  consider the conjunctions of positive {\em and} negative atoms by making use of the results in \Cref{section-satplus} and \Cref{section-satminus}. 
The main challenge is to handle the interactions between positive atoms and
negative atoms. Compared to~\cite[Prop.~11]{ACCFS23}, we rely purely on the satisfiability of an $\Alogic$ formula, and completely get rid of unsatisfiability checks. This is crucial in obtaining the $\NP$-membership. 

\Cref{lemma:combination} below formalizes the conditions under which positive and negative atoms can be jointly satisfied. Intuitively, it combines the results of \Cref{lemma:satplus,lemma:satminus} and extends them with an additional check ensuring that the interaction between positive and negative atoms does not lead to a contradiction. A contradiction may arise, for instance, when formulas such as $\kh(p,q)$, $\kh(q,r)$, and $\neg\kh(p,r)$ are all present simultaneously.
As shown below, this additional check is carried out by means of a partition $\{D,\bar{D}\}$ of the set of pairs of indices of positive atoms, together with the reflexive and transitive closure of~$\bar{D}$, denoted by $\C{\bar{D}}$.

\begin{lemma}\label{lemma:combination}
    Let $I,J$ be finite, disjoint sets of indices and
 let $\varphi, \theta$ be subjective formulas:
    \[
    \varphi = 
            \overbrace{\big(\bigwedge_{i \in I} \kh(\psi_i,\chi_i) \big)}^{\varphi^+}
            \land
            \overbrace{\big(\bigwedge_{j \in J} \lnot\kh(\psi_j,\chi_j)\big)}^{\varphi^{-}}, 
    \]
    \[
    \theta =
        \overbrace{
            \bigwedge_{i \in I} (\A(\lnot \psi_i) \lor \E\chi_i)
        }^{\theta^+}
        \wedge 
        \overbrace{
            \bigwedge_{j \in J} \E(\psi_j \land \lnot\chi_j)
        }^{\theta^{-}}
        \wedge
        \left(
        \bigvee_{D \subseteq (I \times I)} \theta_{D}
        \right)
    \]
    where: 
        \begin{multline*} 
        \theta_D \ = \ 
            \big(\bigwedge_{j \in J}
                \bigwedge_{(s,t) \in \C{\bar{D}}}
                    (\E(\psi_j \land  \lnot \psi_s) \vee
                    \E(\chi_t \land \lnot\chi_j))\big)
            \, \wedge \\ 
            \big(\bigwedge_{(t,s) \in D}
                \E(\chi_{t} {\land} \lnot\psi_{s})\big). 
        \end{multline*} 
    Then, $\varphi$ is satisfiable iff $\theta$ is satisfiable.
\end{lemma}

Observe that $\theta$ is a subjective $\SFive$ formula. Moreover, there exists an equivalent $\SFive$ formula in disjunctive normal form (DNF) in which each disjunct contains a number of $\E$-subformulas cubic in $\card{I \cup J}$ (see also the proof of \Cref{th:satkh-np}).

\begin{proof}
    \setcounter{ieqnnum}{0}%
    $(\Rightarrow)$ Suppose that $\varphi$ is satisfiable:
    there is an $\KHlogic$ model $\model = \tup{\S,\Rfamily,\V}$ with
        \ieqn{\truthset{\model}{\varphi} = \S}.
    The model $\model$ readily gives us an $\SFive$ model $\nmodel = \tup{\S, \V}$.
    We need to show that there is a subset $D \subseteq (I \times I)$ such that $\truthset{\nmodel}{\theta^+ \wedge \theta^{-} \wedge \theta_D} = \S$. 
    By the $\Rightarrow$ direction of the proof of \Cref{lemma:satplus}, we have
    immediately $\truthset{\nmodel}{\theta^+} = \S$.
    Similarly, by the $\Rightarrow$ direction of the proof of \Cref{lemma:satminus}, we get $\truthset{\nmodel}{\theta^{-}} = \S$.
    
    Now, let $D = \setof{(t,s) \in I^2}{\truthset{\model}{\chi_t} \nsubseteq \truthset{\model}{\psi_s}}$.
    The definition of $D$ guarantees $\truthset{\model}{\chi_t \land \lnot\psi_s} \neq \emptyset$, for all $(t,s) \in D$. This establishes $\truthset{\nmodel}{\bigwedge_{(t,s) \in D} \E(\chi_t \land \lnot\psi_s)} = \S$ as required.


    It remains to prove \ieqn{\truthset{\nmodel}{\E(\psi_j \land  \lnot \psi_s) \lor \E(\chi_t \land \lnot\chi_j)} = \S} for all $j \in J$ and all $(s,t) \in \C{\bar{D}}$. 
    We claim that for all $(s,t) \in \C{\bar{D}}$, $\truthset{\model}{\kh(\psi_s,\chi_t)} = \S$. 
    If this claim holds, \ieqnref0 follows directly from the fact that either $\truthset{\model}{\psi_j} \nsubseteq \truthset{\model}{\psi_s}$ or $\truthset{\model}{\chi_t} \nsubseteq \truthset{\model}{\chi_j}$,
    since otherwise, 
    we can use the witness plan for $\truthset{\model}{\kh(\psi_s,\chi_t)} = \S$ to contradict $\truthset{\model}{\kh(\psi_j,\chi_j)} {=} \emptyset$.
    
    We complete the $\Rightarrow$ direction of the lemma by proving the claim.  
    To this end, suppose  $(s,t) \in \C{\bar{D}}$.
    Then, there are $i_0, \ldots, i_k$ in $I$ such that
    $i_0 = s$, $i_k = t$ and for all $0 \leq m < k$,
    $\pair{i_m}{i_{m+1}} \in \bar{D}$. 
    We proceed by induction on $k$.

       \begin{description}
            \item[Base case ($k=0$):]
            This means $s = t$. From  $s \in I$, we conclude immediately $\truthset{\model}{\kh(\psi_s,\chi_s)} = \S$.
    
            \item[Inductive step ($k>0$):]
            By the induction hypothesis, we have \ieqn{\truthset{\model}{\kh(\psi_s,\chi_{i_{k-1}})} = \S}.
            By definition of $\bar{D}$, 
            $\truthset{\model}{\chi_{i_{k-1}}} \subseteq \truthset{\model}{\psi_t}$. 
            Since $t \in I$, we also have \ieqn{\truthset{\model}{\kh(\psi_t,\chi_t)} = \S}.
            Then, let $\pi$ be a plan witnessing \ieqnref{-1}, and $\pi'$ be a plan witnessing \ieqnref0.
            By composition, the plan $\pi\pi'$ witnesses $\truthset{\model}{\kh(\psi_s,\chi_t)} = \S$.
        \end{description}
    Bringing together all the preceding steps, we conclude that the satisfiability of $\varphi$ guarantees the satisfiability of $\theta$.
    

    \noindent
    ($\Leftarrow$) Suppose that $\theta$ is satisfiable.
    Then, exists an $\SFive$ model $\nmodel = \tup{\S, \V}$ such that $\truthset{\nmodel}{\theta} = \S$ (equivalently $\truthset{\nmodel}{\theta} \neq \emptyset$, since $\theta$ is a subjective formula, see \Cref{lemma-subjective-fragment}).
    By definition of $\theta$, there is some $D \subseteq (I \times I)$ 
    s.t. 
    $\truthset{\nmodel}{\theta^+ \wedge \theta^{-} \wedge \theta_D} = \S$.
    
    We use the information encoded in $\theta^+ \wedge \theta^{-} \wedge \theta_D$
    to construct, from $\nmodel$, an $\KHlogic$ model $\model$ that satisfies~$\varphi$.
    The central difficulty in this construction lies in defining the interpretation of the actions.
    On the one hand, we must guarantee the existence of appropriate witness plans for the satisfaction of each atom $\kh(\psi_i, \chi_i)$ with $i \in I$.
    On the other hand, we must simultaneously prevent these plans from becoming witnesses for the satisfaction of any 
    atom $\kh(\psi_j, \chi_j)$ with $j \in J$.

    In working towards the definition of $\model$, let 
    \[
        K = \setof{k \in I}{\truthset{\nmodel}{\psi_k} = \emptyset}.
    \]
    We define $\model = \tup{\S, \Rfamily, \V}$, where
    for all $a_{\ell} \in \ACT$: 
    \begin{align*}
        \R_{a_{\ell}} &=
            \begin{cases}
                {\truthset{\nmodel}{\psi_{\ell}} \times \truthset{\nmodel}{\chi_{\ell}}}
                    & \text{if } \ell \in (I \setminus K) \\
                \emptyset
                    & \text{otherwise.}
            \end{cases}
    \end{align*}
    Note that, for any non-empty subset 
    $Y \subseteq \truthset{\nmodel}{\psi_{\ell}}$ with $\ell \in(I \setminus K)$, we have, by definition, that $\R_{a_\ell}(Y) = \truthset{\nmodel}{\chi_{\ell}} \neq \emptyset$.
    To see why $\truthset{\nmodel}{\chi_{\ell}} \neq \emptyset$,
    it suffices to recall that
    $\truthset{\nmodel}{\A(\lnot \psi_{\ell}) \lor \E\chi_{\ell}}
    = \S$ and $\truthset{\nmodel}{\A(\lnot \psi_{\ell})} = \emptyset$ because $\ell \not \in K$.

    We claim that $\truthset{\model}{\varphi} = \S$. To establish this claim, we must prove the following:  $\truthset{\model}{\kh(\psi_i,\chi_i)} = \S$, for all $i \in I$, and  $\truthset{\model}{\kh(\psi_j,\chi_j)} = \emptyset$, for all $j \in J$.

    Let us show first that, for all $i \in I$, $\truthset{\model}{\kh(\psi_i,\chi_i)} = \S$ by a case analysis, distinguishing between $i \in K$ or~$i \in (I{\setminus}K)$.

    \begin{description} 
        \item[Case $i \in K$:]
        The empty plan~$\varepsilon$ serves as a witness for $\truthset{\model}{\kh(\psi_i,\chi_i)} = \S$. Observe first that $\truthset{\model}{\psi_i} = \emptyset$, then $\truthset{\model}{\psi_i} \subseteq \stexec(\varepsilon) = \S$. Next, note that $\R_{\varepsilon}(\truthset{\model}{\psi_i}) = \truthset{\model}{\psi_i} = \emptyset$, which is trivially contained in $\truthset{\model}{\chi_i}$. These two conditions guarantee $\truthset{\model}{\kh(\psi_i,\chi_i)} = \S$. 

        \item[Case $i \in (I \setminus K)$:]
        In this case, the witness plan is simply the single action $a_i \in \ACT \subseteq \PLANS$. By construction, this action is defined so that two key conditions hold.
        First, the set of states satisfying $\psi_i$ is included in the domain of execution of $a_i$, i.e., $\truthset{\model}{\psi_i} \subseteq \stexec(a_i)$. Second, the states reached by executing $a_i$ from $\truthset{\model}{\psi_i}$ are all contained in $\truthset{\model}{\chi_i}$, i.e., $\R_{a_i}(\truthset{\model}{\psi_i}) \subseteq \truthset{\model}{\chi_i}$.
        Taken together, these conditions establish that $a_i$ is indeed a witness plan,  therefore $\truthset{\model}{\kh(\psi_i,\chi_i)} = \S$.
    \end{description}

    The two cases above tell us that we can construct a suitable witness plan for every $i \in I$. This ensures that 
    $\truthset{\model}{\varphi^+} = \S$. 

    Let us now turn to proving $\truthset{\model}{\kh(\psi_j,\chi_j)} = \emptyset$ for all $j \in J$. 
    To this end, choose an arbitrary $j \in J$.
    Note that, from $\truthset{\nmodel}{\theta^{-}} = \S$, $\truthset{\model}{\psi_j} \neq \emptyset$. 
    To conclude the proof, we must establish the following: for all $\pi \in \PLANS$, if $\truthset{\model}{\psi_j} \subseteq \stexec(\pi)$ then $\R_{\pi}(\truthset{\model}{\psi_j}) \nsubseteq \truthset{\model}{\chi_j}$. We proceed with a case analysis on whether $\pi$ is $\varepsilon$, or whether it is some other plan.

    \begin{description} 
        \item[Case ($\pi = \varepsilon$):]
        Observe that $\truthset{\nmodel}{\E(\psi_j \land \lnot\chi_j)} = \S$ guarantees that there is at least one state in $\truthset{\model}{\psi_j}$ that does not belong to $\truthset{\model}{\chi_j}$. Consequently, 
        $\R_{\varepsilon}(\truthset{\model}{\psi_j}) = \truthset{\model}{\psi_j} \nsubseteq \truthset{\model}{\chi_j}$. This implies that 
        $\varepsilon$ cannot be a witness for 
        the satisfaction of $\kh(\psi_j,\chi_j)$.

        \item[Case ($\pi \neq \varepsilon$):]
        Then $\pi \in \ACT^{+}$ (i.e., $\pi$ is a plan of length at least one). 
        W.l.o.g., we may restrict our attention to plans $\pi \in \setof{a_{\ell}}{\ell \in (I \setminus K)}^{+}$ such that $\truthset{\model}{\psi_j} \subseteq \stexec(\pi)$.
        This is possible for two reasons. 
        On the one hand, if $\truthset{\model}{\psi_j} \nsubseteq \stexec(\pi)$ then $\pi$ cannot be a witness plan for $\kh(\psi_j,\chi_j)$.
        On the other hand, if $\pi$ contains an action $a_\ell$ with $\ell \notin (I \setminus K)$, then $\R_{\pi} = \emptyset$ as $\R_{a_\ell} = \emptyset$. Since $\truthset{\model}{\psi_j} \neq \emptyset$, we have $\truthset{\model}{\psi_j} \nsubseteq \stexec(\pi)$, contradicting the assumption that $\truthset{\model}{\psi_j} \subseteq \stexec(\pi)$.
        
        So, choose $\pi = a_s \dots a_t \in \setof{a_{\ell}}{\ell \in (I \setminus K)}^{+}$ with $\truthset{\model}{\psi_j} \subseteq \stexec(\pi)$.
        We claim that if $\truthset{\model}{\psi_j} \subseteq \stexec(\pi)$, 
        then ${(s,t) \in \C{\bar{D}}}$ and $\R_{\pi}(\truthset{\model}{\psi_s}) = \truthset{\model}{\chi_t}$.
        As a consequence, let us see why $\R_{\pi}(\truthset{\model}{\psi_j}) \subseteq \truthset{\model}{\chi_j}$ cannot hold. 
        By assumption and definition of $\pi$, $\emptyset \neq 
        \truthset{\model}{\psi_j} \subseteq \stexec(\pi) \subseteq \truthset{\model}{\psi_s}$.
        This implies $\truthset{\nmodel}{\E(\psi_j \land  \lnot \psi_s)} = \emptyset$.
        Note that $\truthset{\model}{\theta_D} = \S$ 
        implies \ieqn{\truthset{\nmodel}{\E(\psi_j \land  \lnot \psi_s) \lor \E(\chi_t \land \lnot\chi_j)} = \S}, and, therefore, $\truthset{\nmodel}{\E(\chi_t \land \lnot\chi_j)} = \S$. 
        Consequently, $\truthset{\model}{\chi_t} \nsubseteq \truthset{\model}{\chi_j}$. However by $\R_{\pi}(\truthset{\model}{\psi_s}) = \truthset{\model}{\chi_t}$, $\emptyset \neq 
        \truthset{\model}{\psi_j} \subseteq \truthset{\model}{\psi_s}$ and $\truthset{\model}{\psi_j} \subseteq \stexec(\pi)$, we get 
        $\R_{\pi}(\truthset{\model}{\psi_j})  =
        \truthset{\model}{\chi_t}$ too, and, therefore, 
        $\R_{\pi}(\truthset{\model}{\psi_j}) \nsubseteq 
        \truthset{\model}{\chi_j}$.
        This shows that no $\pi \in \setof{a_{\ell}}{\ell \in (I \setminus K)}^{+}$ that is strongly executable from all states in $\truthset{\model}{\psi_j}$ can serve as a witness for establishing $\truthset{\model}{\kh(\psi_j,\chi_j)} = \S$.
        \cut{
        So, choose $\pi = a_s \dots a_t \in \setof{a_{\ell}}{\ell \in (I \setminus K)}^{+}$ with $\truthset{\model}{\psi_j} \subseteq \stexec(\pi)$.
        {\color{Cyan}
        We claim that if $\truthset{\model}{\psi_j} \subseteq \stexec(\pi)$, then, ${(s,t) \in \C{\bar{D}}}$ and $\R_{\pi}(\truthset{\model}{\psi_s}) = \truthset{\model}{\chi_t}$.
        }%
        If this claim holds, then, $\truthset{\model}{\chi_t} \nsubseteq \truthset{\model}{\chi_j}$.
        To see why, first, note that $\truthset{\model}{\theta_D}$ implies \ieqn{\truthset{\nmodel}{\E(\psi_j \land  \lnot \psi_s) \lor \E(\chi_t \land \lnot\chi_j)} = \S}.
        By assumption and definition of $\pi$, $\truthset{\model}{\psi_j} \subseteq \stexec(\pi) \subseteq \truthset{\model}{\psi_s}$.
        This implies $\truthset{\nmodel}{\E(\psi_j \land  \lnot \psi_s)} = \emptyset$.
        Thus, it must be that $\truthset{\model}{\chi_t} \nsubseteq \truthset{\model}{\chi_j}$.
        Otherwise, \ieqnref0 would be false in $\nmodel$.
        This shows that no $\pi \in \setof{a_{\ell}}{\ell \in (I \setminus K)}^{+}$ that is strongly executable from all the states in $\truthset{\model}{\psi_j}$, can serve as a witness for obtaining ${\truthset{\model}{\kh(\psi_j,\chi_j)} {=} \S}$.
        }
    \end{description}

    Combining the cases $\pi = \varepsilon$ and $\pi \neq \varepsilon$, we can readily see that for every $j \in J$ it is impossible to construct a witness plan for $\kh(\psi_j,\chi_j)$ in $\model$. We conclude that $\truthset{\model}{\varphi^{-}} = \S$.

    We complete the $\Leftarrow$ direction of the lemma by proving the claim: if $\truthset{\model}{\psi_j} \subseteq \stexec(\pi)$, then ${(s,t) \in \C{\bar{D}}}$ and $\R_{\pi}(\truthset{\model}{\psi_s}) = \truthset{\model}{\chi_t}$.
    More precisely, we prove that for all $\pi = a_s \dots a_t \in \setof{a_{\ell}}{\ell \in (I \setminus K)}^{+}$, if $\truthset{\model}{\psi_j} \subseteq \stexec(\pi)$, then $(s,t) \in \C{\bar{D}}$ and $\R_{\pi}(\truthset{\model}{\psi_s}) = \truthset{\model}{\chi_t}$.
    We proceed by induction on the length of $\pi$.

    \begin{description}
    \item[Base case ($|\pi| = 1$):]
    Let $\pi = a_\ell$ with ${\ell \in (I \setminus K)}$.
    Suppose that $\truthset{\model}{\psi_j} \subseteq \stexec(\pi)$.
    Note that, in this case, $s = \ell = t$.
    The reflexivity of $\C{\bar{D}}$ yields ${(s, t) \in \C{\bar{D}}}$.
    Moreover, since $\emptyset \neq \truthset{\model}{\psi_j} \subseteq \truthset{\model}{\psi_s}$, the definition of $\R_{a_\ell}$ ensures $\R_{a_\ell}(\truthset{\model}{\psi_j}) = \truthset{\model}{\chi_t}$.

    \item[Inductive step ($|\pi| > 1$):] Let $\pi' = a_s \dots a_u$, and ${\pi=\pi'a_t}$ with $\pi \in \setof{a_{\ell}}{\ell \in (I \setminus K)}^{+}$.
    Suppose
        $\truthset{\model}{\psi_j} \subseteq \stexec(\pi)$.
    Then, ${\truthset{\model}{\psi_j} \subseteq \stexec(\pi')}$.
    By the induction hypothesis, $(s,u) \in \C{\bar{D}}$ and $\R_{\pi'}(\truthset{\model}{\psi_j}) = \truthset{\model}{\chi_u}$.
    From the assumption that $\truthset{\model}{\psi_j}\subseteq\stexec(\pi)$, we conclude $\truthset{\model}{\chi_u} \subseteq \truthset{\model}{\psi_t}$.
    This means that $(u,t) \notin D$, and so $(u,t) \in \bar{D}$.
    From the transitivity of $\C{\bar{D}}$, $(s,t) \in \C{\bar{D}}$.
    Further, from $\R_{\pi'}(\truthset{\model}{\psi_j}) = \truthset{\model}{\chi_u}$
    with $\truthset{\model}{\chi_u} \neq \emptyset$ (see above), 
    $\truthset{\model}{\chi_u} \subseteq \truthset{\model}{\psi_t}$, and the definition of $\R_{a_t}$, we get $\R_{\pi}(\truthset{\model}{\psi_j}) = \truthset{\model}{\chi_t}$.
    
    \end{description}

    Bringing together all the preceding steps, we conclude that the satisfiability of $\theta$ entails the satisfiability of $\varphi$. 
\end{proof}

\cut{
From the previous developments we can conclude:

\begin{corollary}
    \label{cor:equisat-translation}
    For all conjunctions of atoms in $\KHlogic$, 
    one can compute in exponential-time 
    an $\SFive$ subjective formula $\theta$ such that $\varphi$ is satisfiable iff $\theta$ is satisfiable, $\theta$ is in DNF and each disjunct has polynomial size.  
\end{corollary}
}

\subsection{The satisfiability algorithm}

We have incrementally introduced a translation from a class of formulas in
$\KHlogic$ to formulas in $\SFive$ that preserves satisfiability.
Next, we use this translation to provide an optimal procedure for
solving the satisfiability problem for $\KHlogic$.


\begin{theorem}
    \label{th:satkh-np}
    ${\sf SAT}(\KHlogic)$ is in $\NP$.
\end{theorem}

\begin{proof}
It was already established in~\cite[Prop.~4]{ACCFS23} that nested occurrences of $\kh$ can be eliminated via a preprocessing step.
Here, we refine this approach by exploiting the nondeterministic nature of our procedure, which allows us to obtain a formula in the precise syntactic form required by~\Cref{lemma:combination} directly, and thus eliminates the need for a separate preprocessing stage.
Let $\varphi$ be an {\em arbitrary} $\KHlogic$ formula. To check its satisfiability status, we first build a formula $\varphi^+\wedge\varphi^-$ such that $\varphi^+$ is a conjunction of positive atoms, and
$\varphi^-$ is a conjunction of negative atoms.
More precisely, 
the decision procedure is made of three successive nondeterministic steps, each running in
polynomial-time, followed
by a final model-checking step
that runs in deterministic polynomial-time. 
\begin{description}
\item[(Guess a conjunction of atoms)] From the formula
$\varphi$, we guess a conjunction of atoms as follows. 
Let $\varphi_{cur}$ be an $\KHlogic$ formula,
$\varphi^{+}_{cur}$ be a conjunction of positive atoms,
$\varphi^{-}_{cur}$ be a conjunction of negative atoms, 
and $\kh(\aformulabis, \aformulater)$ be a positive atom occurring
in $\varphi_{cur}$.
The formula $\varphi_{cur} \wedge \varphi^{+}_{cur} \wedge \varphi^{-}_{cur}$
is satisfiable 
iff at least one formula below is satisfiable:
\begin{itemize}
\item $\varphi_{cur}[\top / \kh(\aformulabis, \aformulater)]
\wedge (\varphi^{+}_{cur} \wedge \kh(\aformulabis, \aformulater))
\wedge \varphi^{-}_{cur}$. 
\item $\varphi_{cur}[{\perp}/ \kh(\aformulabis, \aformulater)]
\wedge \varphi^{+}_{cur} 
\wedge (\varphi^{-}_{cur} \wedge \neg \kh(\aformulabis, \aformulater))$. 
\end{itemize}

Above, $\varphi_{cur}[\top / \kh(\aformulabis, \aformulater)]$
is obtained from $\varphi_{cur}$ by replacing the occurrences of
$\kh(\aformulabis, \aformulater)$
by $\top$ (similar reading with $\perp$). This process reduces by
one the number
of positive atoms 
within
$\varphi_{cur}[\top / \kh(\aformulabis, \aformulater)]$
and
$\varphi_{cur}[{\perp}/ \kh(\aformulabis, \aformulater)]$.
Hence, starting from $\varphi_{cur} = \varphi$,
$\varphi^{+}_{cur} = \top$ and $\varphi^{-}_{cur} = \top$, 
we can repeat the process of selecting one alternative
(a linear number of times in $\size{\varphi}$), so that we reach
final values $\varphi_{cur},\varphi^{+}_{cur},\varphi^{-}_{cur}$
such that $\varphi_{cur}$ is a propositional formula.
One can easily show that
$\varphi_{cur} \wedge \varphi^{+}_{cur} \wedge \varphi^{-}_{cur}$ is
satisfiable iff
$\varphi^{+}_{cur} \wedge (\varphi^{-}_{cur} \wedge \neg \kh(\varphi_{cur},
\perp))$. By correctness of the syntactic transformation,
we have the guarantee that our initial formula $\varphi$ 
is satisfiable iff one of the formulas obtained by this iterative
process of the form
$\varphi^{+}_{cur} \wedge (\varphi^{-}_{cur} \wedge \neg \kh(\varphi_{cur},
\perp))$ is satisfiable. 
We write~$\theta$
to denote the $\SFive$ formula, following the
notations from \Cref{lemma:combination}, 
with $\varphi$ 
being the final conjunction of atoms
$\varphi^{+}_{cur} \wedge (\varphi^{-}_{cur} \wedge \neg \kh(\varphi_{cur},
\perp))$ that is guessed in polynomial-time.
Note that the total number of atoms is bounded above
by $\size{\aformula}$ and we do not mean to construct~$\theta$, which
would be prohibitive to get $\NP$-membership.

\item[(Guess a disjunct in the DNF of $\theta$)] Let us guess a disjunct
$\theta^{\dag}$ for some
DNF of $\theta$, which amounts to guess also
some $D \subseteq I \times I$, see the expression for $\theta$
in \Cref{lemma:combination}.
The atoms in $\theta^{\dag}$ can be explicitly
enumerated as follows. 
\begin{itemize}
\item For all $i \in I$, $\theta^{\dag}$ contains
$\A(\lnot \psi_i)$ or $\E\chi_i$
and for all $j \in J$, $\theta^{\dag}$ contains
$\E(\psi_j \land \lnot\chi_j)$.
\item For all $(t,s) \in D$, $\theta^{\dag}$ contains
$\E(\chi_{t} \land \lnot\psi_{s})$.
\item For all $j \in J$ and $(s,t) \in \C{\bar{D}}$, $\theta^{\dag}$ contains
$\E(\psi_j \land  \lnot \psi_s)$ or $\E(\chi_t \land \lnot\chi_j)$.
\end{itemize}
Since $\card{I \uplus J} \leq \size{\varphi}$ and
$\max(\card{D},\card{\C{\bar{D}}}) \leq \card{I}^2$, we conclude that
$\theta^{\dag}$ contains at most $3 \size{\varphi}^3$ atoms
(rough upper bound) and its size is polynomial in $\size{\varphi}$. 
\item[(Guess a small model)] By~\cite[Lemma~6.1]{Ladner77}, any satisfiable $\SFive$
conjunction $\aformulabis$ made of
$\E$-formulas and $\A$-formulas has a model
$\nmodel = \tup{\S,\V}$ such that $\card{\S} \leq n + 1$
 where $n$ is the number of $\E$-formulas in
$\aformulabis$. Guess a model $\nmodel^{\dag} = \tup{\S,\V}$ with
$\card{\S} \leq (3 \size{\varphi}^3)+1$. 
\end{description}
Since the model-checking problem for S5 is in $\Poly$, checking whether
$\truthset{\nmodel^{\dag}}{\theta^{\dag}} \neq \emptyset$ can also be performed
in time polynomial in $\size{\varphi}$, whence the $\NP$-membership.
\end{proof}

From the $\NP$-hardness of propositional logic we obtain: 

\begin{corollary}
    ${\sf SAT}(\KHlogic)$ is $\NP$-complete.
\end{corollary}

We conclude this section with an example illustrating how to translate $\KHlogic$ formulas to finally invoke an $\Alogic$ solver.

\begin{example}
Let $\varphi$ be $\kh(p\wedge q, r\wedge t) \vee \kh(p,r)$. We apply our procedure step by step:
\begin{enumerate}
    \item Replace $\kh(p\wedge q, r\wedge t)$ by $\top$ and $\kh(p,r)$ by $\bot$. Thus, we get: $\varphi_{cur} = \top \vee \bot$, $\varphi^{+}_{cur} = \kh(p\wedge q, r\wedge t)$, and $\varphi^{-}_{cur} = \neg \kh(p,r)$.
    \item We have then $D=\set{(1,1)}$ and $\C{\bar{D}}= \set{(1,1)}$. The formula $\theta$ will contain two disjuncts, which is satisfiable, and hence, so is $\varphi$.
\end{enumerate}
\end{example}

\subsection{Small model property}
As a consequence of the proof of \Cref{th:satkh-np},
we can also establish the following polynomial-size model property.

\begin{corollary}\label{corollary-quadratic-size-model-property}
For every satisfiable $\KHlogic$ formula $\aformula$, there is 
$\model = \tup{\S,\Rfamily,\V}$ 
s.t. 
$\truthset{\model}{\aformula} \neq \emptyset$ and
$\card{\S}$ is in $\mathcal{O}(\size{\aformula}^3)$. 
\end{corollary} 

It is worth noting that \Cref{corollary-quadratic-size-model-property} alone is
not sufficient to obtain $\NP$ membership, since the model-checking problem
for $\KHlogic$ is $\PSPACE$-complete~\cite[Th.~1]{DF23}.
Nevertheless, we do not exclude that additional constraints on the models
arising from our previous proofs could yield a subclass of models for which
model checking is in $\Poly$, just as the one-step model-checking for models obtained via bisimulation contractions in~\cite{AFM26}. This would allow one to recover $\NP$
membership via this route.

%% file: final.tex
We characterized the exact complexity of the satisfiability problem for the
seminal knowing-how logic $\KHlogic$~\cite{Wang15lori}.
As a by-product, we obtained a polynomial-size model property for the logic,
which had remained elusive until now. 
Our main result establishes that the satisfiability problem is $\NP$-complete,
improving on the $\Sigma^P_2$ upper bound given in~\cite{ACCFS23}.
We streamline and simplify the argument by relying on an
equisatisfiability-preserving translation into the modal logic S5.
Surprisingly, our proof shows that witness plans and strong executability play
no role in satisfiability.
This observation explains the gap between the $\NP$ bound for satisfiability
and the $\PSPACE$-completeness of the model-checking problem~\cite{DF23}.


As future work, it would be interesting to investigate whether our proof
strategy can be extended to characterize the complexity of other knowing-how
logics, such as those incorporating intermediate constraints
(see, e.g.,~\cite{LiWang17}). In addition, the key ideas developed in our proof readily translate into an effective satisfiability-checking procedure for $\KHlogic$. In this respect, the next step is to provide an implementation for automated reasoning with $\KHlogic$, together with an empirical evaluation to evaluate its practical applicability.
